\documentclass[graybox]{svmult}

\usepackage{mathptmx, amsfonts, amsmath}       % selects Times Roman as basic font
\usepackage{helvet}         % selects Helvetica as sans-serif font
\usepackage{courier}        % selects Courier as typewriter font
\usepackage{type1cm}        % activate if the above 3 fonts are
\usepackage{makeidx}         % allows index generation
\usepackage{graphicx}        % standard LaTeX graphics tool
\usepackage{multicol}        % used for the two-column index
\usepackage[bottom]{footmisc}% places footnotes at page bottom
\usepackage{bm}

\def\N{\mathbb{N}}

\def\R{\mathbb{R}}
\let\e=\varepsilon

\let\.=\cdot
\let\0=\emptyset

\def\tr{\text{\rm tr}}

\def\square{\hbox{$\sqcap\kern-7pt\sqcup$}}

\def\be{\begin{equation}}
\def\ee{\end{equation}}
\def\bea{\begin{eqnarray}}
\def\eea{\end{eqnarray}}

\DeclareMathAlphabet{\pazocal}{OMS}{zplm}{m}{n}

\makeindex             % used for the subject index
\begin{document}

\title*{From the Hartree to the Vlasov dynamics:\\ conditional strong convergence }
% Use \titlerunning{Short Title} for an abbreviated version of
% your contribution title if the original one is too long

\author{Chiara Saffirio}
% Use \authorrunning{Short Title} for an abbreviated version of
% your contribution title if the original one is too long

\institute{Chiara Saffirio \at 
Departement Mathematik und Informatik, Universit\"at Basel, Spiegelgasse 1, CH-4051 Basel\\
\email{chiara.saffirio@unibas.ch}
%\and Name of Second Author \at Name, Address of Institute \email{name@email.address}
}

%
% Use the package "url.sty" to avoid
% problems with special characters
% used in your e-mail or web address
%
\maketitle

%\abstract*{ }

\abstract{We review the recent results \cite{S19,S-VP} concerning the semiclassical limit from the Hartree dynamics to the Vlasov equation with singular potentials and extend them to the case of more general radial interactions. We prove that,  at positive temperature, the Hartree dynamics converges  in trace norm to the Vlasov one, for a particular class of initial states.}


\section{Introduction}
\label{sec:intro}

\medskip

{\bf The Hartree dynamics.} We consider the nonlinear Hartree-Fock equation 
\be\label{eq:HF}
i\,\hbar\,\partial_t\,\omega_{N,t}=[\,\mathcal{H}_{\rm HF}\,,\,\omega_{N,t}]\,,
\ee
where $\omega_{N,t}$ is a one-particle fermionic operator, i.e. a nonnegative trace class operator over $L^2(\R^3)$ such that $\tr\ \omega_{N,t}=N$, $0\leq\omega_{N,t}\leq 1$ and $N=\hbar^{-3}$.  In \eqref{eq:HF}  $[\mathcal{A},\mathcal{B}]$ denotes the commutator between the operator $\mathcal{A}$ and the operator $\mathcal{B}$,  $\mathcal{H}_{\rm HF}$  is the Hamilton operator given by
\be\label{eq:Ham}
\mathcal{H}_{\rm HF}=-\hbar^2\Delta+({V}*\varrho_t)-\mathcal{X}_t\,,
\ee
where $\Delta$ is the Laplace operator, ${V}:\R^3\to\R$ is a two-body interaction potential, $\varrho_t:\R^3\to\R$  is the density of fermions in the position space at time $t$, i.e. for $x\in\R^3$ $\varrho_t(x)=N^{-1}\omega_{N,t}(x;x)$, and $\mathcal{X}_t$ is the exchange operator defined through its kernel $\mathcal{X}_t(x;y)=N^{-1}{V}(x-y)\omega_{N,t}(x;y)$, where we have used the notation $\mathcal{A}(\cdot;\cdot)$ to denote the kernel of the operator $\mathcal{A}$. We notice that the Hartree-Fock equation \eqref{eq:HF} propagates in time the fermionic structure of the operator $\omega_{N}$. Namely if $\omega_N$ is a fermionic operator its evolution at time $t$ according to the Hartree-Fock equation $\omega_{N,t}$ is also a fermionic operator, i.e. it satisfies the properties $\tr\ \omega_{N,t}=N$ and $0\leq\omega_{N,t}\leq 1$.\\
It has been shown in \cite{EESY,BPS13,BJPSS,PRSS,S18,BGGM1,BGGM2,FK,P,PP,BBPPT} that the Hartree-Fock equation is the mean-field approximation of the $N$-body fermionic Schr\"odinger dynamics for a class of quasi-free states exhibiting a semiclassical structure (see \cite{BPS-rev} for further readings). In particular, as first observed in \cite{EESY}, the mean-field scaling for a system of interacting fermions is naturally coupled with a semiclassical scaling, being the Planck constant $\hbar$ proportional to $N^{-\frac{1}{3}}$, where $N$ is the number of fermions. \\
Since the energy contribution associated with the exchange term can be shown to be subleading (see \cite{BPS13,S-VP}), we will from now on drop the exchange term and consider the fermionic Hartree equation
\be\label{eq:H}
i\,\hbar\,\partial_t\,\omega_{N,t}=[\mathcal{H}_{\rm H}\,,\,\omega_{N,t}]\,,
\ee
where $\mathcal{H}_{\rm H}$ is the Hamilton operator $\mathcal{H}_{\rm HF}$ defined in \eqref{eq:Ham} with exchange term $\mathcal{X}_t=0$. 
\medskip

{\bf The Vlasov equation.} The Hartree equation \eqref{eq:H} is $N$ dependent because of the relation $\hbar=N^{-\frac{1}{3}}$ and the choice $\tr\ \omega_{N,t}=N$. As we perform the limit $N\to\infty$, we are simultaneously performing the semiclassical limit $\hbar\to 0$. As first established by Narnhofer and Sewell in \cite{NS}, the dynamics of $N$ interacting fermions converges in the mean-field and semiclassical regime to a solution to the Vlasov equation
\be\label{eq:Vl}
\partial_t \widetilde{W}_t+v\cdot\nabla_x\widetilde{W}_t+(\nabla V*\widetilde{\rho}_t)\cdot\nabla_v\widetilde{W}_t=0\,,
\ee
where, for each $t\in\R_+$, $\widetilde{W}_t:\R^3\times\R^3\to\R$ is a probability density on the phase space representing the probability that a particle at time $t$ is in position $x\in\R^3$ with velocity $v\in\R^3$. The function $\widetilde{\rho}_t:\R^3\to\R$ is the spatial density of particles defined for each $t$ as $\widetilde{\rho}_t(x)=\int_{\R^3}\widetilde{W}_t(x,v)\,dv$ and $V:\R^3\to\R$ is a two-body interaction potential. \\
The Cauchy problem associated with Eq. \eqref{eq:Vl} has been largely investigated. In \cite{D79} Dobrushin proved wellposedness in the case of potentials satisfying $V\in\mathcal{C}_c^2(\R^3)$. When the interaction is Coulomb or gravitational, existence of classical solutions under regularity assumptions on the initial datum has been established in \cite{I} and \cite{OU} respectively in one and two dimensions. The three dimensional case has been addressed in \cite{BD} for small initial data and in a more general setting by Pfaffelmoser in \cite{Pf} and by Lions and Perthame in \cite{LP}.   \\
The goal of this paper is to review recent results on the strong convergence of the Hartree equation \eqref{eq:H} towards the Vlasov equation \eqref{eq:Vl} when the interaction potential is singular. 
\medskip

{\bf Weyl and Wigner transforms.} To study the semiclassical limit of Eq. \eqref{eq:H}, we observe that the solution of the Hartree equation \eqref{eq:H} is an operator on $L^2(\R^3)$, whereas the solution to the Vlasov equation \eqref{eq:Vl} is a function on the phase space $\R^3\times\R^3$.
 To compare them, we need to set the problem either on the space of operators on $L^2(\R^3)$ or on the space of functions on $\R^3\times\R^3$. To this end, we introduce the notion of Wigner transform. We define the Wigner transform of the one-particle operator $\omega_{N,t}$ with kernel $\omega_{N,t}(\cdot\,;\,\cdot)$ as
 \be\label{eq:Wigner}
 W_{N,t}(x,v)=\left(\dfrac{\hbar}{2\pi}\right)^3\int_{\R^3} \omega_{N,t}\left(x+\dfrac{\hbar y}{2};x-\dfrac{\hbar y}{2}\right)\,e^{-iv\cdot y}dy\,.
 \ee
 Thus, for each $t>0$, $W_{N,t}:\R^3\times\R^3\to\R$ and it is normalized, indeed $$\int_{\R^3\times\R^3}W_{N,t}(x,v)dxdv=\hbar^3\tr\ \omega_{N,t}=\hbar^3\,N=1\,,$$ where in the last identity we used the relation $\hbar=N^{-1/3}$. Despite $W_{N,t}$ is normalized, in general it is not a probability density as it is not nonnegative.\\
The inverse transformation is called Weyl quantization. Given a function $W_{N,t}$ on the phase space, we define the Weyl quantization of $W_{N,t}$ by 
\be\label{eq:Weyl}
\omega_{N,t}(x;y)=N\int W_{N,t}\left(\dfrac{x+y}{2},v\right)e^{iv\cdot(x-y)/\hbar}dv\,,
\ee
where $\omega_{N,t}(x;y)$ is the kernel associated with the one-particle operator $\omega_{N,t}$. Therefore, the Weyl quantization transforms a function on the phase space into the kernel of a one-particle operator. Moreover, the quantity $\int_{\R^3}W_{N,t}(x,v)dv$ describes the density of fermions in position space at point $x\in\R^3$ and $\int_{\R^3}W_{N,t}(x,v)dx$ represents the density of particles with velocity $v\in\R^3$.

\medskip

{\bf Initial states.} We are interested in solutions to Eq. \eqref{eq:H} which are evolutions of initial states describing equilibrium states of trapped systems. In the mean-field regime such equilibrium states are expected to be approximately quasi-free. In this paper we are interested in the evolution of quasi-free states at positive temperature, usually referred to as mixed states. More precisely, Shale-Stinespring condition (see Theorem 9.5 in \cite{Lieb}) guarantees that every one-particle operator $\omega_N$ such that $\tr\ \omega_N=N$ and $0\leq\omega_N\leq 1$ is the one-particle reduced density of a quasi-free state with $N$ fermions. When looking at the semiclassical limit, the advantage of considering mixed quasi-free states is that the associated Wigner transform is a regular function on the phase space. This is in general not the case when looking at zero temperature states (the so called pure states).  See \cite{BPSS} for further reading on this matter.  
\medskip

{\bf State of art.} The problem of obtaining the Vlasov equation as a classical limit of a dynamics of interacting quantum particles has been largely investigated. The first result in this direction is due to Narnhofer and Sewell. In \cite{NS} the authors consider a system of $N$ fermions interacting through a two-body analytic potential $V\in\mathcal{C}^{\omega}(\R^3)$ and they obtain the Vlasov equation directly from the dynamics given by the $N$-fermion Schr\"odinger equation. In \cite{Spohn81} the assumption on the interaction potential has been relaxed to $V\in\mathcal{C}^2(\R^3)$. The case of $N$ bosons given by WKB states in the mean-field regime combined with a semiclassical limit has been analysed in \cite{GMP}. \\
The semiclassical limit from the Hartree dynamics towards the Vlasov equation has been proven in weak topology for regular and singular interactions (here included the Coulomb potential) in \cite{LionsPaul,MM,FLP}. This analysis has been extended to the Hartree-Fock equation in \cite{GIMS}.\\
The above cited references establish weak convergence towards the solution to the Vlasov equation, but do not provide any control on the rate of convergence. The problem of obtaining explicit bounds on the convergence rate has been first addressed in \cite{APPP} where the convergence from the Hartree equation to the Vlasov dynamics has been established in strong topology for regular interactions and later extended in \cite{PP,AKN1,AKN2,BPSS} and in \cite{S19,S-VP} for inverse power law potentials (included Coulomb interaction).\\
We also mention that a new approach has been initiated in \cite{GMP1, GolsePaul1}, where a notion of pseudo-distance reminiscent of the Monge-Kantorovich distance for classical probability measures has been introduced. Explicit bounds on the rate of convergence in the topology induced by such quantum analogous of the Monge-Kantorovich distance have been obtained for smooth (Cf. \cite{GolsePaul1}) and singular potentials, including the Coulomb interaction (Cf. \cite{Lafleche1,Lafleche2}).    \\
More recently, in the context of regular interactions, the joint mean-field and semiclassical limit for the dynamics of $N$ fermions has been studied also in \cite{CLL19} by means of BBGKY type of hierarchy for the $k$-particle Husimi measure.  
\medskip

{\bf Main result.} In this paper we are concerned with the strong convergence of the Hartree dynamics towards the Vlasov equation with interactions satisfying the following assumptions: 
\begin{itemize}
\item[a) ] let $V:\R^3\to\R$ be a radially symmetric potential that is three times differentiable away from $x=0$ and denote $$V^{(m)}(|x|):=\dfrac{d^{m}}{d|x|^m}V(x)\,;$$ 
\item[b) ] for $0\leq m\leq 3$, assume 
$$\lim_{|x|\to\infty}|x|^mV^{(m)}(|x|)=0\,;$$
\item[c) ]\ for $\delta\in(0,3/2]$, we assume
$$\int_0^k |r^2V^{(3)}(r)-rV^{(2)}(r)|r^{\frac{9}{2}-\delta}dr<\infty\,. $$
\end{itemize}
We can think for instance to inverse power law interaction potentials, but more general interactions satisfying the above assumptions are included in the analysis performed in this paper. \\
For some $a>0$ and $k\in\N$, we introduce the weighed Sobolev spaces  $H_a^k$, defined as the space of all square integrable functions $f$ on the phase space $\R^3\times\R^3$ such that the following norm is finite
$$
\|f\|_{H_a^k}:=\left(\sum_{|\beta|\leq k}\int (1+x^2+v^2)^a\,|\nabla^\beta f(x,v)|^2\,dx\,dv\right)^{\frac{1}{2}}
$$
where $\beta$ is a multi-index and $\nabla^\beta$ can act on both $x$ and $v$ variables.\\
We are now ready to state our result.
\begin{theorem}\label{thm:main}
Let $V$ be a two-body potential satisfying assumptions {\rm a)}, {\rm b)} and {\rm c)} above and $\hbar=N^{-\frac{1}{3}}$. Let $\omega_{N}$ a sequence of fermionic operators, i.e. $\tr\ \omega_N=N$ and $0\leq \omega_N\leq N$, and denote by $\omega_{N,t}$ the solution to the Hartree equation \eqref{eq:H} with initial data $\omega_N$. Let $W_N$ the Wigner transform of the operator $\omega_N$ and assume\footnote{Such a solution indeed exists if the initial data are sufficiently regular. For the precise assumptions, see \cite{S19,S-VP}.} that there exists a unique smooth solution $\widetilde{W}_{N,t}$ of the Vlasov equation \eqref{eq:Vl} with initial data $W_N$ such that $\nabla^2\rho_t\in L^{\infty}(\R^3)$ for all $t\in[0,T]$ and $\|W_N\|_{H_4}^k$ is bounded uniformly in $N$ for $k=1,\dots, 6$.\\ Assume moreover that there exists a time $T>0$ and a constant $C>0$ such that
\be\label{eq:semiclassical-assump}
\sup_{t\in[0,T]}\sum_{i=1}^3\left[ \|\varrho_{|[x_i\,,\,\widetilde{\omega}_{N,t}]|}\|_{1}+\|\varrho_{|[x_i\,,\,\widetilde{\omega}_{N,t}]|}\|_{\infty} \right]\leq C\,N\,\hbar
\ee
where 
\be\label{eq:rho-def}
\rho_{|[x_i,\widetilde{\omega}_{N,t}]|}(x):=|[x_i\,,\,\widetilde{\omega}_{N,t}]|(x;x)\,.
\ee
Then there exist constants $C_k$, $k=0,\dots,4$, depending only on $T$ and on the $H_4^{2+k}$ norm of the Wigner transform of the initial data $W_N$ such that
\be\label{eq:trace-norm}
\tr\ \left| \omega_{N,t}-\widetilde{\omega}_{N,t}\right|\leq C_0\,N\,\hbar\,\left[ 1+C_1\hbar+C_2\hbar^2+C_3\hbar^3+C_4\hbar^4 \right]\,. 
\ee
\end{theorem}
\begin{remark}
We recall that $\tr\ \omega_{N,t}=N$, hence the bound  \eqref{eq:trace-norm} is non-trivial, showing that the Vlasov equation is a good semiclassical approximation for the fermionic dynamics given by the Hartree equation with singular interaction potential $V$. 
\end{remark}
\begin{remark}
We observe that the exchange term in \eqref{eq:HF} is subleading in the limit of $N$ large, i.e. as $\hbar\to 0$. For this reason the bound \eqref{eq:trace-norm} is expected to remain correct if we consider the Hartree-Fock equation \eqref{eq:HF} instead of the Hartree equation \eqref{eq:H}.
\end{remark}
\begin{remark}
Hypothesis \eqref{eq:semiclassical-assump} is a stringent assumption. At the moment we are not able to prove the bound \eqref{eq:semiclassical-assump} for the solution of the Vlasov equation. Nevertheless, there is a special situation in which the assumption is satisfied, that is for translation invariant states of the Vlasov equation \eqref{eq:Vl} and for regular steady solutions of the Vlasov equation \eqref{eq:Vl} when $V$ is the gravitational potential. See \cite{PRSS, S-VP} for extended explanations and examples.    
\end{remark}
\medskip

{\bf Strategy of the proof.}
The idea of the proof is to obtain a Gr\"{o}nwall type inequality to compare the solution $\omega_{N,t}$ of the Hartree equation \eqref{eq:H} with a solution $\widetilde{\omega}_{N,t}$ to the Weyl trasformed Vlasov equation, i.e. at the operator level and not as functions on the phase space. The operators $\omega_{N,t}$ and $\widetilde{\omega}_{N,t}$ are fermionic operators, i.e. they satisfy $0\leq\omega_{N,t}\leq 1$ and $0\leq\widetilde{\omega}_{N,t}\leq 1$, with the normalisation $\tr\ \omega_{N,t}=\tr\ \widetilde{\omega}_{N,t}=N$.\\ 
In the Gr\"{o}nwall type inequality we look for there are two terms appearing: the first one is the dominant term that will allow us to close the inequality; the second one is an subleding term and it will determine the rate of convergence of $\omega_{N,t}$ towards $\widetilde{\omega}_{N,t}$, as $\hbar=N^{-1/3}\to 0$. In the dominant term, the Laplace operator related to the kinetic part of the Hartree Hamiltonian \eqref{eq:Ham} appears. To deal with it, we make use of a unitary transformation that acts as a change the reference frame, canceling the kinetic factor. Moreover, we deal with the singularity of the potential $V$ by using a generalised version of the Fefferman and de la Llave representation formula. For general radially symmetric and fast decreasing potentials such a formula was provided in \cite{HS} (see Proposition \ref{prop:FDLL}). This is the key tool in the proof of Theorem \ref{thm:main}, as it allows to isolate the singularity of the potential $V$ at zero and close the Gr\"{o}nwall inequality. To this end, we need to control the trace norm of the commutator between the Weyl transformed solution $\widetilde{\omega}_{N,t}$ of the Vlasov equation and the multiplication operator by a Gaussian $\chi_{(r,z)}(x)=\exp(-|x-z|^2/r^2)$ (see Lemma \ref{lemma:trace}). This is the point at which the quantities \eqref{eq:semiclassical-assump}-\eqref{eq:rho-def} appear and the main reason why we have to restrict to initial data satisfying the bound \eqref{eq:semiclassical-assump} at time $t>0$. To bound the error term, and therefore determine the convergence rate, we use again the generalized Fefferman - de la Llave representation formula for rapidly decreasing, radially symmetric interactions.
\medskip

The paper is organised as follows: in Sect.~\ref{sect:preliminary} we present some auxiliary results, such as the precise statement of the generalized Fefferman - de la Llave representation formula and the key estimate on the commutator $\tr\ |[\chi_{(r,z)},\widetilde{\omega}_{N,t}]|$; in Sect.~\ref{sect:proof} we give the main steps to prove Theorem \ref{thm:main}, focusing on the parts of the proof that differ from the Coulomb interaction case treated in \cite{S-VP}.

%%%%%%%%%%%%%%%%%%%%%%%%%%%%%%%%%%%%%%%%%%%%%
%%%%%%%%%%%%%%%%%%%%%%%%%%%%%%%%%%%%%%%%%%%%%

\section{Preliminary estimates}
\label{sect:preliminary}

\begin{lemma}\label{lem:est-tr}
Let $\omega_{N,t}$ be a solution to the Hartree equation \eqref{eq:H} with initial datum $\omega_N$.
Denote by $W_N$ the Wigner transform of $\omega_N$ and let $\widetilde{W}_{N,t}$ be the solution of the Vlasov equation \eqref{eq:Vl} with initial data $W_N$.
Denote by
$\widetilde{\omega}_{N,t}$ the Weyl transform of $\widetilde{W}_{N,t}$. Then
\be\label{eq:est-tr}
\tr\ |\omega_{N,t}-\widetilde{\omega}_{N,t}|\leq\frac{1}{\hbar}\int_0^t \tr\ \left|\left[ V*(\rho_s-\widetilde{\rho}_s),\tilde{\omega}_{N,s} \right]\right|\,ds +\frac{1}{\hbar}\int_0^t \tr\ |B_s|\,ds
\ee
where, for every $s\in [0,t]$, $B_s$ is the operator with kernel
\be\label{eq:def-B}
B_s(x;y)=\left[\left(V*\widetilde{\rho}_s\right)(x)-\left(V*\widetilde{\rho}_s\right)(y)-\nabla\left(V*\widetilde{\rho}_s\right)\left(\frac{x+y}{2}\right)\cdot(x-y)\right]\,\widetilde{\omega}_{N,t}(x;y)\,.
\ee
\end{lemma}
%%%%%%%%%%%%%%%%%%%%%
\begin{proof}
Applying the Weyl quantization to the Vlasov equation \eqref{eq:Vl}, we obtain the following equation for the operator $\widetilde{\omega}_{N,t}$: 
\be\label{eq:vlasov-weyl}
i\,\hbar\,\partial_t\,\widetilde{\omega}_{N,t}=[-\hbar^2\,\Delta,\widetilde{\omega}_{N,t}]+\mathcal{A}_t
\ee
where $\mathcal{A}_t$ is the operator with integral kernel
\begin{equation*}
\mathcal{A}_t(x;y)=\nabla\left(V*\widetilde{\rho}_t\right)\left(\frac{x+y}{2}\right)\cdot(x-y)\,\widetilde{\omega}_{N,t}(x;y)\,.
\end{equation*}
We introduce $\mathcal{U}(t;s)$, the two-parameter group of unitary transformations generated by the Hartree Hamiltonian $\mathcal{H}_{\rm H}(t):=-\hbar^2\Delta+V*\rho_t$
\begin{equation*}
\left\{
\begin{array}{l}
i\,\hbar\,\partial_s\,\mathcal{U}(t;s)=\mathcal{H}_{\rm H}(t)\,\mathcal{U}(t;s)\\
\mathcal{U}(s;s)=1 
\end{array}
\right.
\end{equation*}
so that $\omega_{N,t}=\mathcal{U}(t;0)\,\omega_{N}\,\mathcal{U}^*(t;0)$. The role of $\mathcal{U}(t;s)$ is to cancel the kinetic part of the Hamiltonian $\mathcal{H}_{\rm H}(t)$ by conjugating the difference between $\omega_{N,t}$ and $\widetilde{\omega}_{N,t}$ with $\mathcal{U}(t;s)$ and performing the time derivative. This leads to
\begin{equation*}
\begin{split}
i\,\hbar\,\partial_s\,(\mathcal{U}^*(t;s)\,&(\omega_{N,s}-\widetilde{\omega}_{N,s})\,\mathcal{U}(t;s))\\&=\mathcal{U}^*(t;s)\,[\mathcal{H}_{\rm H}(s),\omega_{N,s}-\widetilde{\omega}_{N,s}]\,\mathcal{U}(t;s)\\
&+\mathcal{U}^*(t;s)\,([\mathcal{H}_{\rm H}(s),\omega_{N,s}]-[-\hbar^2\Delta,\widetilde{\omega}_{N,s}]-\mathcal{A}_s)\,\mathcal{U}(t;s)\\
&=\mathcal{U}^*(t;s)\,\left(\left[V*\rho_s,\widetilde{\omega}_{N,t} \right]-\mathcal{A}_s\right)\,\mathcal{U}(t;s)\\
&=\mathcal{U}^*(t;s)\,\left(\left[V*(\rho_s-\widetilde{\rho}_s),\widetilde{\omega}_{N,s}\right]+\mathcal{B}_s\right)\,\mathcal{U}(t;s)\,
\end{split}
\end{equation*}
where $\mathcal{B}_s$ is the operator with integral kernel \eqref{eq:def-B}. Recalling that $\widetilde{\omega}_{N,0}=\omega_N$, Duhamel's formula yields
 \begin{equation}\label{eq:omega-omegatilde}
 \begin{split}
\mathcal{U}^*(t;s)\,(\omega_{N,s}-\widetilde{\omega}_{N,s})\,\mathcal{U}(t;s)&=\frac{1}{i\,\hbar}\int_0^t \mathcal{U}^*(t;s)\,\left[V*(\rho_s-\widetilde{\rho}_s),\widetilde{\omega}_{N,s}\right]\,\mathcal{U}(t;s)\,ds\\
&+\frac{1}{i\,\hbar}\int_0^t \mathcal{U}^*(t;s)\,\mathcal{B}_s\,\mathcal{U}(t;s)\,ds
\end{split}
 \end{equation}
 The bound \eqref{eq:est-tr} follows by taking the trace norm in the above expression and using that $\mathcal{U}(t;s)$ is a family of unitary operators.  
\end{proof}

\begin{proposition}\label{prop:1term-tr}
Under the same assumptions of Lemma \ref{lem:est-tr} and Theorem \ref{thm:main}, there exists a constant $C>0$ such that
\be\label{eq:1term-tr}
%\begin{split}
\tr\ \left|\left[V*(\rho_s-\widetilde{\rho}_s)\,,\,\widetilde{\omega}_{N,s}\right]\right|\leq C\,\hbar\,\tr\ |\omega_{N,s}-\widetilde{\omega}_{N,s}| 
%\\
%&+C_2\,N\,\e+ C_3\,N\,\e^{\frac{7}{5}+\frac{3}{5}\delta}+C_4\,N\,\e^{\frac{17}{5}+\frac{3}{5}\delta}\\
%&+C_5\,N\,\e^{\frac{23}{10}}+C_6\,N\,\e^{\frac{14}{5}}+C_7\,N\,\e^{\frac{43}{10}}+C_8\,N\,\e^{\frac{24}{5}}
%\end{split}
\ee
\end{proposition}

The proof of Proposition \ref{prop:1term-tr} relies on the generalization proved in \cite{HS} of the Fefferman - de la Llave representation formula established in \cite{FDLL} for the Coulomb potential. 

\begin{proposition}[Theorem 1 in \cite{HS}]\label{prop:FDLL}
For $n\geq 2$, let $V:\R^n\to\R$ be a radial function that is $[n/2]+2$ times differentiable away from $x=0$, where $[a]$ denotes the integer part of $a$. For $m\in\N_0$ denote $V^{(m)}(|x|)=d^m/d|x|^m\,V(x)$. Assume that $\lim_{|x|\to\infty}|x|^mV^{(m)}(|x|)=0$ for all $0\leq m\leq [n/2]+1$. Then
\be\label{eq:FDLL}
V(x)=\int_0^\infty \int_{\R^n} g(r)\,{\bf 1}_{\{|x-z|\leq r\}}{\bf 1}_{\{|z|\leq r\}}\,dz\,dr
\ee
where 
\begin{equation}\label{eq:g(r)}
\begin{split}
g(r)=\frac{(-1)^{[\frac{n}{2}]}}{\Gamma\left(\frac{n-1}{2}\right)}\frac{2}{(\pi\,r^2)^{\frac{(n-1)}{2}}}&\left( \int_r^\infty ds\,V^{([\frac{n}{2}]+2)}(s)\,\left(\frac{d}{ds}\right)^{n-1-[\frac{n}{2}]}s(s^2-r^2)^{\frac{1}{2}(n-3)}\right.\\
&+\left.{\bf 1}_{\{n=2k+1,\, k\in\N\}}V^{([\frac{n}{2}]+2)}(r)\,r(2r)^{\frac{1}{2}(n-3)}\Gamma\left(\frac{n-1}{2}\right)\right)
\end{split}
\end{equation}
where $\Gamma(\cdot)$ denotes the Gamma function.
\end{proposition}
\begin{remark}
We observe that the Coulomb potential $V(x)=\dfrac{1}{|x|}$ satisfies the assumptions in Proposition \ref{prop:FDLL}. In that case $g(r)$ has a simple expression and the Fefferman - de la Llave representation formula writes 
\begin{equation*}
\dfrac{1}{|x-y|}=\dfrac{1}{\pi}\int_0^\infty \int_{\R^n} \dfrac{1}{r^5}\,{\bf 1}_{\{|x-z|\leq r\}}\,{\bf 1}_{\{|y-z|\leq r\}}\,dz\,dr
\end{equation*}
\end{remark}
\begin{remark}
We can easily replace the characteristic function $\bm{1}_{\{|x-z|\leq r\}}$ by a smooth function varying on the same scale at the price of having a different constant in front of the expression. In what follows, we choose to use the Gaussian 
\be\label{eq:gaussian}
\chi_{(r\,,\,z)}(x):=\exp\left(- \frac{|x-z|^2}{r^2} \right)\,.
\ee
Hence, with the notations introduced in Proposition \ref{prop:FDLL}, the generalised Fefferman - de la Llave representation formula reads
\be\label{eq:GFDLL}
V(x)=\dfrac{4}{\pi}\int_0^\infty \int_{\R^n} g(r)\,\chi_{(r,z)}(x)\,\chi_{(r,z)}(0)\,dz\,dr
\ee
Moreover, integrating out the $z$ variable, we get
\be\label{eq:V-short}
V(x)=\dfrac{4}{\pi}\int_0^\infty  r^3\,g(r)\,\chi_{(r/\sqrt{2},x)}(0)\,dr
\ee

\end{remark}

For later use, we recall here the definition of the Hardy-Littlewood maximal function:

\begin{definition}\label{def:HLMax}
For  $z\in\R^3$ and $B$ ball in $\R^3$ centred at zero, the Hardy-Littlewood maximal function of a function $f$ is defined as 
\[
f^*(z)=\sup_{B\,:\,z\in B}\dfrac{1}{|B|}\int_B f(x)\,dx
\]
\end{definition}

We can now state the following Lemma, which gives a bound in trace norm on the commutator of $\widetilde{\omega}_{N,t}$ and the multiplication operator $\chi_{(r,z)}$.

\begin{lemma}[Lemma 3.1 in \cite{PRSS}]\label{lemma:trace}
Let $\chi_{(r\,,\,z)}(x)$ be as in \eqref{eq:gaussian} and assume $\left[ x_i ,\widetilde{\omega}_{N,t} \right]$ to be  trace class for all $t \in [0;T]$ and $i=1,2,3$. Then, for all $\delta\in(0,1/2)$, there exists a positive constant  $C$ such that 
the bound
\be\label{eq:trace-com}
\tr\ |[\chi_{(r,z)}\,,\,\widetilde{\omega}_{N,t}]|\leq C\,r^{\frac{3}{2}-3\delta}\sum_{i=1}^3 \|\rho_{|[x_i\,,\,\widetilde{\omega}_{N,t}]|}\|_{L^1}^{\frac{1}{6}+\delta}\left(\rho^*_{|[x_i\,,\,\widetilde{\omega}_{N,t}]|}(z) \right)^{\frac{5}{6}-\delta}
\ee
holds pointwise, where $\rho_{|[x_i\,,\,\widetilde{\omega}_{N,t}]|}$ is defined in \eqref{eq:rho-def} and $\rho^*_{|[x_i\,,\,\widetilde{\omega}_{N,t}]|}$ denotes the Hardy-Littlewood maximal function of $\rho_{|[x_i\,,\,\widetilde{\omega}_{N,t}]|}$ introduced in Definition \ref{def:HLMax}.
\end{lemma}
\begin{proof}
We consider the integral kernel of the commutator $[\chi_{(r,z)}, \widetilde{\omega}_{N,t}]$ and write it as 
\[ \begin{split}
[ \chi_{(r\,,\,z)} , \widetilde{\omega}_{N,t}] (x;y) &= \left( \chi_{(r\,,\,z)}(x) -\chi_{(r\,,\,z)}(y) \right) \widetilde{\omega}_{N,t}(x;y)  \\
&= \int_0^1 d s\;\dfrac{d}{d s} e^{-\frac{(x-z)^2}{r^2}s}\widetilde{\omega}_{N,t} (x;y)e^{-\frac{(x-z)^2}{r^2}(1-s)}  \\
&=  - \int_0^1 d s\;e^{-s\frac{(x-z)^2}{r^2}} \left[ \frac{(x-z)^2}{r^2}, \widetilde{\omega}_{N,t} \right](x;y) \, e^{-(1-s)\frac{(y-z)^2}{r^2}}\,. 
\end{split} \]
We can therefore write the commutator as
\be\label{eq:comm1}
\begin{split}
[ &\chi_{(r\,,\,z)} , \widetilde{\omega}_{N,t}] 
\\ %&= - \int_0^1 d s\; \chi_{(r/\sqrt{s},z)} (x) \left[ \frac{(x-z)^2}{r^2}, \widetilde{\omega}_{N,t} \right] \chi_{(r/\sqrt{1-s}\,,\,z)} (x) \\  
&= - \sum_{i=1}^3 \int_0^1 d s\; \chi_{(r/\sqrt{s}\,,\,z)} (x)  \left( \frac{(x-z)_i}{r^2}\left[ x_i , \widetilde{\omega}_{N,t} \right] + \left[ x_i , \widetilde{\omega}_{N,t} \right] \frac{(x-z)_i}{r^2} \right) \chi_{(r/\sqrt{1-s}\,,\,z)}(x) \\
&= \sum_{i=1}^3 \mathfrak{A}^{(1)}_i + \mathfrak{A}^{(2)}_i \, .
\end{split}
\end{equation}
where, with an abuse of notation, we denote by $\chi_{(.,.)} (x)$ both the function of $x$ and the corresponding multiplication operator. 

We focus on the first term on the r.h.s. of (\ref{eq:comm1}) and fix $i=1$. The other components of the first term, and the three components of the second term can then be treated similarly. By the spectral decomposition of the commutator $\left[ x_1 ,\widetilde{\omega}_{N,t} \right]$ we have 
\[ [x_1 ,\widetilde{\omega}_{N,t}] = i \sum_j \lambda_j | \varphi_j \rangle \langle \varphi_j| \]
for a sequence of eigenvalues $\lambda_j \in \mathbf{R}$ and an orthonormal system $\varphi_j$ in $L^2 (\mathbf{R}^3)$ (we introduced $i=\sqrt{-1}$ on the r.h.s., because the commutator is anti self-adjoint). We find 
\[ \begin{split} 
\mathfrak{A}_1^{(1)} &= -\int_0^1 d s\; \chi_{(r/\sqrt{s},z)}(x) \frac{(x-z)_1}{r^2} \left[ x_1, \widetilde{\omega}_{N,t} \right]\chi_{(r/\sqrt{1-s},z)}(x)\\
&= -\frac{i}{r} \sum_j \lambda_j \int_0^1 \frac{d s}{\sqrt{s}} \; \bigg\vert \, \chi_{(r/\sqrt{s}\,,\,z)} (x) \frac{(x-z)_1}{r/\sqrt{s}} \varphi_j \bigg\rangle \bigg\langle \chi_{(r/\sqrt{1-s}\,,\,z)} (x) \varphi_j \bigg\vert\,. 
\end{split} \]
Using that $\tr\ \left||\psi_1\rangle\langle\psi_2|\right|=\|\psi_1\|\,\|\psi_2\|$ we obtain
\begin{equation}\label{eq:tr1-1} \begin{split} 
\tr\ | \mathfrak{A}_1^{(1)} | &\leq \frac{1}{r} \sum_j |\lambda_j| \int_0^1 \frac{ds}{\sqrt{s}} \left\| \chi_{(r/\sqrt{s}\,,\,z)} (x) \frac{|x-z|}{r/\sqrt{s}} \varphi_j \right\| \, \left\|  \chi_{(r/\sqrt{1-s}\,,\,z)} (x) \varphi_j \right\|  \\ 
&\leq \frac{1}{r} \int_0^1 \frac{ds}{\sqrt{s}} \left( \sum_j |\lambda_j| \left\| \chi_{(r/\sqrt{s}\,,\,z)} (x) \frac{|x-z|}{r/\sqrt{s}} \varphi_j \right\|^2 \right)^{1/2} \\ &\hspace{4cm}\times \left( \sum_j |\lambda_j| \left\| \chi_{(r/\sqrt{1-s}\,,\,z)} (x) \varphi_j \right\|^2 \right)^{1/2}\,.\end{split} \end{equation}
We compute
\begin{equation}\label{eq:max1} \begin{split} 
\sum_j |\lambda_j| \left\| \chi_{(r/\sqrt{1-s},z)} (x) \varphi_j \right\|^2  &= \int dx \, e^{-2(1-s)(x-z)^2/r^2} \rho_{|[x,\widetilde{\omega}_{N,t}]|} (x)\\ &\leq C \frac{r^3}{(1-s)^{3/2}} \,  \rho^*_{|[x,\widetilde{\omega}_{N,t}]|} (z) \end{split} \end{equation} 
where $\rho^*_{|[x_i,\widetilde{\omega}_{N,t}]|}$ is the Hardy-Littlewood maximal function associated with $\rho_{|[x_i,\widetilde{\omega}_{N,t}]|}$. To prove (\ref{eq:max1}), we write
\[ \begin{split} e^{-2(1-s)(x-z)^2/r^2} &= \int_0^1 \chi (t \leq e^{-2(1-s)(x-z)^2/r^2}) dt \\ &= \int_0^1 \chi \left( |x-z| \leq \sqrt{\frac{r^2 \log (1/t)}{2(1-s)}} \right) dt  \end{split} \]
and, using Fubini, we find 
\[\begin{split} \int dx \, &e^{-2(1-s)(x-z)^2/r^2} \rho_{|[x,\widetilde{\omega}_{N,t}]|} (x) \\ &= \int_0^1 dt \int dx \, \chi \left( |x-z| \leq \sqrt{\frac{r^2 \log (1/t)}{2(1-s)}} \right) \,
\rho_{|[x,\widetilde{\omega}_{N,t}]|} (x) \\ &\leq C \frac{r^3}{(1-s)^{3/2}} \,  \rho^*_{|[x,\widetilde{\omega}_{N,t}]|} (z)  \int_0^1  (\log (1/t))^{3/2} \\ &\leq C \frac{r^3}{(1-s)^{3/2}} \,  \rho^*_{|[x,\widetilde{\omega}_{N,t}]|} (z) \end{split} \]
which shows (\ref{eq:max1}). Similarly to (\ref{eq:max1}), we also find 
\[ \sum_j |\lambda_j| \left\| \chi_{(r/\sqrt{s},z)} (x) \frac{|x-z|}{r/\sqrt{s}} \varphi_j \right\|^2 \leq C \frac{r^3}{s^{3/2}} \rho^*_{|[x_1,\widetilde{\omega}_{N,t}]|} (z)\,. \]
Combining this bound with the simpler estimate
\[ \sum_j |\lambda_j| \left\| \chi_{(r/\sqrt{s},z)} (x) \frac{|x-z|}{r/\sqrt{s}} \varphi_j \right\|^2 \leq C\sum_j |\lambda_j|  = \| \rho_{|[x_1,\widetilde{\omega}_{N,t}]|} \|_1 \]
we obtain 
\[ \sum_j |\lambda_j| \left\| \chi_{(r/\sqrt{s},z)} (x) \frac{|x-z|}{r/\sqrt{s}} \varphi_j \right\|^2 \leq
C \frac{r^{3\gamma} \, \| \rho_{|[x_1, \widetilde{\omega}_{N,t}]|} \|^{1-\gamma}_1 }{s^{3\gamma/2}} \, \left( \rho^*_{|[x_1, \widetilde{\omega}_{N,t}]|} (z) \right)^{\gamma} \]
for any $0 \leq \gamma \leq 1$. Inserting the last bound and (\ref{eq:max1}) on the r.h.s. of (\ref{eq:tr1-1}) we conclude  
\[ \begin{split} \tr\ | \mathfrak{A}_1^{(1)} | &\leq C r^{(1 + 3\gamma)/2} \| \rho_{|[x_1, \widetilde{\omega}_{N,t}]|} \|_1^{(1-\gamma)/2} \left( \rho^*_{|[x_1, \widetilde{\omega}_{N,t}]|} (z) \right)^{(1+\gamma)/2} \\ &\hspace{4cm} \times \int_0^1 ds\frac{1}{s^{1/2+3\gamma/4} (1-s)^{3/4}}\,. \end{split} \]
Hence, for all $\delta > 0$ we find (putting $\gamma = 2/3 -2\delta$)
\[ \tr\ | \mathfrak{A}_1^{(1)} | \leq  C r^{3/2 - 3\delta} \| \rho_{|[x_1, \widetilde{\omega}_{N,t}]|} \|_1^{1/6+\delta} \left( \rho^*_{|[x_1, \widetilde{\omega}_{N,t}]|} (z) \right)^{5/6 -\delta}
\]
which concludes the proof.
\end{proof}

\begin{proof}[Proof of Proposition \ref{prop:1term-tr}]
We write explicitly the convolution appearing on the l.h.s. of \eqref{eq:1term-tr} and get the following expression for the commutator
\[\label{eq:convolution}
[V*(\varrho_s-\widetilde{\varrho}_s)\,,\,\widetilde{\omega}_{N,s}]=\int (\varrho_s(z)-\widetilde{\varrho}_s(z))\,[V(\cdot-z)\,,\,\widetilde{\omega}_{N,s}]\,dz\;.
\]
By \eqref{eq:GFDLL} and \eqref{eq:gaussian}, we can rewrite the potential $V$ as
\[\label{eq:V-short}
V(x-z)=\dfrac{4}{\pi}\int_0^\infty r^3\,g(r)\,\chi_{(\sqrt{2}r,z)}(x)\,dr\,.
\]
Plugging \eqref{eq:V-short} into \eqref{eq:convolution} and taking the trace norm of \eqref{eq:convolution}, we obtain the bound
\[\begin{split}
\tr\ |[V*&(\varrho_s-\widetilde{\varrho}_s)\,,\,\widetilde{\omega}_{N,s}]|\\
&\leq C\int |\varrho_s(z)-\widetilde{\varrho}_s(z)|\int_0^\infty r^3\,|g(r)|\,\tr\ |[\chi_{(\sqrt{2}r,z)}\,,\,\widetilde{\omega}_{N,s}]|\,dr\,dz
\end{split}
\]
We choose $k>0$ and split the integral in the $r$ variable into two parts: the set $r\in[0,k)$ and the set $r\in[k,\infty)$.\\
For $r\in[0,k)$, we use Lemma \ref{lemma:trace} and get
\[
\begin{split}
\int_0^k r^3\,|g(r)|\,&\tr\ |[\chi_{(\sqrt{2}r,z)}\,,\,\widetilde{\omega}_{N,s}]|\,dr\\
&\leq\int_0^k r^{\frac{9}{2}-3\delta}|g(r)|\sum_{i=1}^3\|\rho_{|[x_i,\widetilde{\omega}_{N,s}]|}\|_{L^1}^{\frac{1}{6}+\delta}\left(\rho^*_{|[x_i,\widetilde{\omega}_{N,s}]|}(z) \right)^{\frac{5}{6}-\delta}\,dr\\
&\leq \int_0^k r^{\frac{9}{2}-3\delta}|g(r)|\sum_{i=1}^3\|\rho_{|[x_i,\widetilde{\omega}_{N,s}]|}\|_{L^1}^{\frac{1}{6}+\delta}\|\rho^*_{|[x_i,\widetilde{\omega}_{N,s}]|}\|_{L^\infty}^{\frac{5}{6}-\delta}\,dr\\
&\leq C\,\sum_{i=1}^3\|\rho_{|[x_i,\widetilde{\omega}_{N,s}]|}\|_{L^1}^{\frac{1}{6}+\delta}\|\rho_{|[x_i,\widetilde{\omega}_{N,s}]|}\|^{\frac{5}{6}-\delta}_{L^\infty}
\end{split}
\]
where in the last inequality we used that $\|\rho^*_{|[x_i,\widetilde{\omega}_{N,s}]|}\|_{L^\infty}\leq\|\rho_{|[x_i,\widetilde{\omega}_{N,s}]|}\|_{L^\infty}$ and that the integral in $r$ converges by assumption c) on the potential $V$.\\
As for $r\in[k,\infty)$, following the same lines of the proof of Lemma \ref{lemma:trace} and choosing the parameter $\gamma=0$, we simply bound the trace norm of the commutator $[\chi_{(\sqrt{2}r,z)},\widetilde{\omega}_{N,s}]$ as
\[
\tr\ |[\chi_{(\sqrt{2}r,z)},\widetilde{\omega}_{N,s}]|\leq C\|\rho_{|[x_i,\widetilde{\omega}_{N,s}]|}\|_{L^1}\,.
\]
Then we are left with the integral in $r$. Using assumption b) on the potential $V$, we conclude that
\[
\int_k^\infty r^3\,|g(r)|\,dr=C\int_k^\infty \left|\,r\,V^{(2)}(r)-r^2\,V^{(3)}(r)\,\right|\,dr<+\infty\,.
\]
Hence
\[
\begin{split}
\tr\ |[V*&(\varrho_s-\widetilde{\varrho}_s)\,,\,\widetilde{\omega}_{N,s}]|\\
&\leq C\int |\varrho_s(z)-\widetilde{\varrho}_s(z)|\sum_{i=1}^3\left( \|\rho_{|[x_i,\widetilde{\omega}_{N,s}]|}\|_{L^1}^{\frac{1}{6}+\delta}\|\rho_{|[x_i,\widetilde{\omega}_{N,s}]|}\|_{L^\infty}^{\frac{5}{6}-\delta}+\|\rho_{|[x_i,\widetilde{\omega}_{N,s}]|}\|_{L^1} \right)\,.
\end{split}
\]
Moreover, using the dual definition of $L^1$ norm, we get
\[\begin{split}
\int |\varrho_s(z)-\widetilde{\varrho}_s(z)|\,dz&=\sup_{\substack{\mathcal{O}\in L^\infty\\ \|\mathcal{O}\|_{\infty}\leq 1 }}\left|\int \mathcal{O}(z)\,(\varrho_s(z)-\widetilde{\varrho}_s(z))\,dz\right|\\
&\leq \dfrac{1}{N}\sup_{\substack{\mathcal{O}\in\mathfrak{B}\\ \|\mathcal{O}\|\leq 1 }}\left|\tr\ \mathcal{O}\,(\omega_{N,s}-\widetilde{\omega}_{N,s})\right|\\
&= \dfrac{1}{N}\,\tr\ |\omega_{N,s}-\widetilde{\omega}_{N,s}|
\end{split}
\]
where $\mathfrak{B}$ denotes the set of bounded operators and $\|\cdot\|$ the operator norm. This concludes the proof.
\end{proof}

\begin{proposition}\label{prop:2term-tr}
Let $\mathcal{B}_t$ be the operator associated with the \eqref{eq:def-B}. Then, there exists a constant $C>0$ depending on $\|\widetilde{W}_{N,t}\|_{H^2_4}$, $\|\nabla^2\widetilde{\rho}_t\|$ such that
\be
\tr\ |\mathcal{B}_t|\leq C\,N\,\hbar^2\,\left(1+\sum_{k=1}^4 \hbar^k\|\widetilde{W}_{N,t}\|_{H_4^{k+2}}\right)\,.
\ee
\end{proposition}
The proof can be found in \cite{S18} and it is based on the following procedure: we write the identity operator 
as
\begin{equation*}
\mathbf{1}=(1-\hbar^2\Delta)^{-1} (1+x^2)^{-1}(1+x^2) (1-\hbar^2 \Delta).
\end{equation*} 
By Cauchy-Schwarz inequality we have
\be\label{eq:tr-B} 
\tr \, |\mathcal{B}_t| \leq \| (1-\hbar^2\Delta)^{-1} (1+x^2)^{-1} \|_{\rm HS} \, \| (1+x^2) (1-\hbar^2 \Delta) \mathcal{B}_t \|_{\rm HS}
\ee
We notice that for some $C>0$ the following bound holds
\begin{equation*}
\| (1-\hbar^2\Delta)^{-1} (1+x^2)^{-1} \|_{\rm HS} \leq C \sqrt{N} 
\end{equation*}
where we have used the explicit form of the kernel of the operator $(1-\hbar^2\Delta)^{-1}$ and the fact that $\hbar^3=N$.\\
We denote  by $U_s$ the convolution of the interaction with the spatial density at time $s$
\be\label{eq:convolution}
U_s:= V* \widetilde{\varrho}_s.
\ee
We introduce the notation $$\widetilde{\mathcal{B}} := (1-\hbar^2 \Delta) \mathcal{B}_s$$ and observe that the kernel of $\widetilde{\mathcal{B}}$ reads 
\be\label{eq:termsB}
\widetilde{\mathcal{B}} (x;x'):=\sum_{j=1}^7 \widetilde{\mathcal{B}}_j (x;x')  
\ee
where
{\small
\begin{align*}
\widetilde{\mathcal{B}}_1 &(x;x')
 = N \left[U_t(x) - U_t(x') - \nabla U_t\left(\dfrac{x+x'}{2}\right)\cdot(x-x')\right]\int \widetilde{W}_{N,t}\left(\dfrac{x+x'}{2},v\right)e^{i\,v\cdot\frac{(x-x')}{\hbar}}dv \displaybreak[0]  \\
\widetilde{\mathcal{B}}_2\ &(x;x')\\
 =-&N\hbar^2\left[\Delta U_t(x) - \dfrac{1}{4}\Delta \nabla U_t\left(\dfrac{x+x'}{2}\right)\cdot(x-x') - \dfrac{1}{2} \Delta U_t \left(\dfrac{x+x'}{2}\right)\right] \int \widetilde{W}_{N,t}\left(\dfrac{x+x'}{2},v\right)e^{i\,v\cdot\frac{(x-x')}{\e}}dv \displaybreak[0]\\
\widetilde{\mathcal{B}}_3\ &(x;x')\\
 =&- \frac{N\hbar^2}{4} \left[U_t(x) - U_t(x') - \nabla U_t\left(\dfrac{x+x'}{2}\right)\cdot(x-x')\right] \int (\Delta_1 \widetilde{W}_{N,t}) \left(\dfrac{x+x'}{2},v\right)e^{i\,v\cdot\frac{(x-x')}{\hbar}}dv \displaybreak[0]\\
\widetilde{\mathcal{B}}_4 &(x;x')
 = N \left[U_t(x) - U_t(x') - \nabla U_t\left(\dfrac{x+x'}{2}\right)\cdot(x-x')\right] \int \widetilde{W}_{N,t}\left(\dfrac{x+x'}{2},v\right) v^2 e^{i\,v\cdot \frac{(x-x')}{\hbar}}dv \displaybreak[0]\\
\widetilde{\mathcal{B}}_5\ &(x;x')\\
 =- &\dfrac{N\hbar^2}{2} \left[\nabla U_t(x) - \dfrac{1}{2} \nabla^2 U_t\left(\dfrac{x+x'}{2}\right) (x-x') - \nabla U_t\left(\dfrac{x+x'}{2}\right)\right] \int (\nabla_1 \widetilde{W}_{N,t}) \left(\dfrac{x+x'}{2},v\right) e^{i\,v\cdot \frac{(x-x')}{\hbar}}dv \displaybreak[0]\\
\widetilde{\mathcal{B}}_6 &(x;x')\\
 = &- N\hbar \left[\nabla U_t(x) - \dfrac{1}{2} \nabla^2 U_t\left(\dfrac{x+x'}{2}\right)(x-x') - \nabla U_t\left(\dfrac{x+x'}{2}\right)\right] \int \widetilde{W}_{N,t}\left(\dfrac{x+x'}{2},v\right) v e^{i\,v\cdot \frac{(x-x')}{\hbar}}dv \displaybreak[0]\\
\widetilde{\mathcal{B}}_7 &(x;x')
 = - N\hbar \left[U_t(x) - U_t(x') - \nabla U_t\left(\dfrac{x+x'}{2}\right)\cdot(x-x') \right] \int (v\cdot \nabla_1 \widetilde{W}_{N,t}) \left(\dfrac{x+x'}{2},v\right) e^{i\,v\cdot \frac{(x-x')}{\hbar}}dv
\end{align*}}
where we used the notation $\nabla_1$ and $\Delta_1$ to indicate derivatives with respect to the first variable.

In order to gain extra powers of $\hbar$, we write  
\begin{equation*}
\begin{split} 
&U_t (x) - U_t (x') - \nabla U_t \left( \frac{x+x'}{2} \right) \cdot (x-x') \\ 
&= \int_0^1 d\lambda \left[ \nabla U_t \left(\lambda x + (1-\lambda) x'\right) - \nabla U_t \left(\dfrac{(x+x')}{2}\right) \right] \cdot (x-x')  \\ 
&= \sum_{i,j=1}^3\int_0^1 d\lambda\left(\lambda-\frac{1}{2}\right)\int_0^1 d\mu\,\partial_i \partial_j U_t\left(h(\lambda, \mu, x,x')\right) (x-x')_i (x-x')_j\,,
\end{split} 
\end{equation*}
where $h(\lambda, \mu, x,x'):=\mu (\lambda x + (1-\lambda) x') + (1-\mu)\dfrac{(x+x')}{2}$.\\
We notice that $U_t$ defined in \eqref{eq:convolution} has a convolution structure. Therefore   derivatives of $U_t$ are equivalent to derivatives of the spatial density $\widetilde{\varrho}_t$. Hence, Fefferman - de la Llave representation formula \eqref{eq:GFDLL} leads to
\be\label{eq:U-taylor}
\begin{split}
U_s &(x) - U_s (x') - \nabla U_s \left( \dfrac{x+x'}{2} \right) \cdot (x-x') \\
&= \sum_{i,j=1}^3\int_0^1 d\lambda\,\left(\lambda-\dfrac{1}{2}\right) \int_0^1 d\mu\,\\
&\quad\int_0^\infty r^{3}g(r) \int dy\,\chi_{(r,y)}\left(h(\lambda, \mu,x,x')\right)\partial_i\partial_j\widetilde{\varrho_s}(y)\, (x-x')_i (x-x')_j\,.
\end{split}
\ee

Plugging \eqref{eq:U-taylor} into the definition of $\widetilde{\mathcal{B}}_1$ and using twice the identity
\be\label{eq:v-der}
(x-x')\int\widetilde{W}_{N,t}\left(\frac{x+x'}{2},v\right)\,e^{iv\cdot\frac{x-x'}{\hbar}}\,dv=-i\hbar\int \nabla_v\widetilde{W}_{N,t}\left(\frac{x+x'}{2},v\right)\,dv 
\ee
and Young's inequality, we get
{\small
\begin{equation*}
\begin{split}
|\widetilde{\mathcal{B}}_1&(x;x')|\\
\leq & C\,N\,\hbar^2\sum_{i,j=1}^3\int_0^1d\lambda\,\left|\lambda-\frac{1}{2}\right|\int_0^1 d\mu\left| \int_0^\infty r^{3}g(r) \right.\\
\quad&\left.\int dy\,\chi_{(r,y)}(h(\lambda,\mu,x,x'))\partial^2_{i,j}\widetilde{\varrho}_s(y)\int dv\,\partial^2_{v_i,v_j}\widetilde{W}_{N,s}\left(\frac{x+x'}{2},v\right)e^{iv\cdot\frac{(x-x')}{\hbar}} \right|
\end{split}
\end{equation*}}

Therefore, the Hilbert-Schmidt norm of the operator $(1+x^2)\widetilde{\mathcal{B}}_1$, where $(1+x^2)$ is the multiplication operator, can be estimated as follows:
\begin{equation*}
\begin{split}
\|(1+x^2)&\widetilde{\mathcal{B}}_1\|_{\rm HS}^2\\
=& N\hbar^4\int dq\int dp'\left[1+q^2+\hbar^2 p^2\right]^2\left|\sum_{i,j=1}^3\int_0^1d\lambda\,\left(\lambda-\frac{1}{2}\right)\int_0^1 d\mu \int_0^\infty r^{3}g(r)\right.\\
\quad&\left.\int dy\,\chi_{(r,y)}(q+\hbar\mu(\lambda-1/2)p)\partial^2_{v_i,v_j}\widetilde{\varrho}_s(y)\int dv\,\partial^2_{v_i,v_j}\widetilde{W}_{N,s}\left(q,v\right)e^{iv\cdot p} \right|^2
\end{split}
\end{equation*}
where we performed the change of variables 
\be\label{eq:change-var}
q=\dfrac{x+x'}{2}\,,\quad\quad\quad p=\dfrac{x-x'}{\hbar}
\ee
with Jacobian $J=8\,\hbar^3=8\,N^{-1}$.\\
%We introduce the shorthand notation
%%
%\be\label{eq:braket-not}
%\langle X \rangle_\e:=1+X^2+\e^2 (X')^2
%\ee
%
%\be
%
%\ee
%
We fix $k>0$ and divide the integral into the two sets  $\pazocal{A}_{<}:=\{ r\in\R_+\ |\ r\leq k\}$ and $\pazocal{A}_>:=\{r\in\R_+\ |\ r>k\}$, so that
\begin{equation}\label{eq:r-split}
\begin{split}
\|(1+&x^2)\widetilde{\mathcal{B}}_1\|_{\rm HS}^2\\
\leq& CN\hbar^4\int dq\int dp[1+q^2+\hbar^2 p^2]^2\sum_{i,j=1}^3\int_0^1d\lambda\,\left|\lambda-\frac{1}{2}\right|\int_0^1 d\mu\\
&\ \left|\int_{\pazocal{A}_<} r^{3}g(r) \int dy\,\chi_{(r,y)}(q+\hbar\mu(\lambda-1/2)p)\partial^2_{v_i,v_j}\widetilde{\varrho}_s(y)\int dv\,\partial^2_{v_i,v_j}\widetilde{W}_{N,s}\left(q,v\right)e^{iv\cdot p} \right|^2\\
+& CN\hbar^4\int dq\int dp[1+q^2+\hbar^2 p^2]^2\sum_{i,j=1}^3\int_0^1d\lambda\,\left|\lambda-\frac{1}{2}\right|\int_0^1 d\mu \\
&\ \left|\int_{\pazocal{A}_>} r^{3}g(r) \int dy\,\chi_{(r,y)}(q+\hbar\mu(\lambda-1/2)p)\partial^2_{v_i,v_j}\widetilde{\varrho}_s(y)\int dv\,\partial^2_{v_i,v_j}\widetilde{W}_{N,s}\left(q,v\right)e^{iv\cdot p} \right|^2
\end{split}
\end{equation}
Denote by $\mathfrak{A}_<$ and $\mathfrak{A}_>$ the first and the second term of the sum on the r.h.s. of \eqref{eq:r-split} respectively. 
 For $\mathfrak{A}_<$ we use Young inequality and H\"{o}lder inequality with conjugated exponents $\theta=1$ and $\theta'=\infty$, then we perform the integral in the $y$ variable to extract $r^3$ which cancels the singularity at zero in the expression for $g(r)$ (Cf. \eqref{eq:g(r)}) together with assumption c) on $V$, thus leading to the bound 
\be\label{eq:HS-B1-r0}
\mathfrak{A}_<\leq CN\hbar^4\int dq\int dv(1+q^2)^2|\nabla^2_v\widetilde{W}_{N,s}(q,v)|^2+CN\hbar^8\int dq\int dv\,|\nabla^4_v\widetilde{W}_{N,s}(q,v)|^2
\ee
where $C$ depends on $\|\nabla^2\widetilde{\varrho}_s\|_{L^\infty}$.\\
For $\mathfrak{A}_>$, we integrate by parts twice in the $y$ variable and recall that  $e^{-|z-y|^2/r^2}(1+|z-y|^2/r^2)$ is bounded uniformly in $z\in\R^3$. Since $\widetilde{\varrho}_s\in L^1(\R^3)$ we get the bound
\be\label{eq:HS-B1-rinf}
\mathfrak{A}_>\leq CN\hbar^4\int dq\int dv(1+q^2)^2|\nabla^2_v\widetilde{W}_{N,s}(q,v)|^2+CN\hbar^8\int dq\int dv\,|\nabla^4_v\widetilde{W}_{N,s}(q,v)|^2
\ee
where $C$ depends on $\|\widetilde{\varrho}_s\|_{L^1}$ and we have used the expression \eqref{eq:g(r)} for $g(r)$ and assumption b) for $V$.\\
Whence, the estimates \eqref{eq:HS-B1-r0} and \eqref{eq:HS-B1-rinf} lead to
\be\label{eq:bound-HS-tildeB1}
\|(1+x^2)\widetilde{\mathcal{B}}_1\|_{\rm HS}
\leq C\sqrt{N}\hbar^2\|\widetilde{W}_{N,s}\|_{H_2^2}
+C\sqrt{N}\hbar^4\|\widetilde{W}_{N,s}\|_{H^4}
\ee
where $C=C(\|\widetilde{\varrho}_s\|_{L^1},\|\nabla^2\widetilde{\varrho}_s\|_{L^\infty})$.\\
The Hilbert-Schmidt norms $\|(1+x^2)\widetilde{\mathcal{B}}_3\|_{\rm HS}$, $\|(1+x^2)\widetilde{\mathcal{B}}_4\|_{\rm HS}$ and $\|(1+x^2)\widetilde{\mathcal{B}}_7\|_{\rm HS}$ can be handled analogously, thus obtaining 
\be
\|(1+x^2)\widetilde{\mathcal{B}}_3\|_{\rm HS}\leq C\sqrt{N}\hbar^4\|\widetilde{W}_{N,s}\|_{H_4^4}
+C\sqrt{N}\hbar^6\|\widetilde{W}_{N,s}\|_{H_4^6}
\ee
\be
\|(1+x^2)\widetilde{\mathcal{B}}_4\|_{\rm HS}\leq C\sqrt{N}\hbar^2\|\widetilde{W}_{N,s}\|_{H_4^2}
+C\sqrt{N}\hbar^4\|\widetilde{W}_{N,s}\|_{H_4^4}
\ee
\be
\|(1+x^2)\widetilde{\mathcal{B}}_7\|_{\rm HS}\leq C\sqrt{N}\hbar^3\|\widetilde{W}_{N,s}\|_{H_3^2}
+C\sqrt{N}\hbar^5\|\widetilde{W}_{N,s}\|_{H_2^5}
\ee
To bound the  $\widetilde{\mathcal{B}}_i$ terms, $i=2,5,6$, in which higher order derivatives of $U$ appear, we proceed as for $\widetilde{\mathcal{B}}_1$ and obtain the following bounds:
\be\label{eq:bound-HS-tildeB2}
\|(1+x^2)\widetilde{\mathcal{B}}_2\|_{\rm HS}
\leq C\sqrt{N}\hbar^4\|\widetilde{W}_{N,s}\|_{H_2^2}
+C\sqrt{N}\hbar^6\|\widetilde{W}_{N,s}\|_{H^4}
\ee
\be\label{eq:bound-HS-tildeB5}
\|(1+x^2)\widetilde{\mathcal{B}}_5\|_{\rm HS}
\leq C\sqrt{N}\hbar^4\|\widetilde{W}_{N,s}\|_{H_2^4}
+C\sqrt{N}\hbar^6\|\widetilde{W}_{N,s}\|_{H^6}
\ee
\be\label{eq:bound-HS-tildeB6}
\|(1+x^2)\widetilde{\mathcal{B}}_6\|_{\rm HS}
\leq C\sqrt{N}\hbar^3\|\widetilde{W}_{N,s}\|_{H_2^2}
+C\sqrt{N}\hbar^5\|\widetilde{W}_{N,s}\|_{H^4}
\ee
where $C=C(\|\widetilde{\varrho}_s\|_{L^1},\|\nabla^2\widetilde{\varrho}_s\|_{L^\infty})$.
We refer to \cite{S-VP} for a detailed proof.

Gathering together all the terms, we get
\be
\begin{split} 
\| (1 + x^{2}) &\widetilde{\mathcal{B}} \|_\text{HS}\\
 \leq C \sqrt{N}& \left[ \hbar^2 \| \widetilde{W}_{N,s} \|_{H^2_{4}} + \hbar^3  \| \widetilde{W}_{N,s} \|_{H^3_{4}} + \hbar^4 \| \widetilde{W}_{N,s} \|_{H^4_{4}} + \hbar^5 \| \widetilde{W}_{N,s} \|_{H^5_4}+\hbar^6 \| \widetilde{W}_{N,s} \|_{H^6_4} \right] 
\end{split}
\ee
%

%%%%%%%%%%%%%%%%%%%%%%%%%%%%%%%%%%%%%%%%%%%%%
%%%%%%%%%%%%%%%%%%%%%%%%%%%%%%%%%%%%%%%%%%%%%

\section{Proof of Theorem \ref{thm:main}}
\label{sect:proof}

From Lemma \ref{lem:est-tr} we have
\be\label{eq:gronwall}
\tr\ |\omega_{N,t}-\widetilde{\omega}_{N,t}|\leq \dfrac{1}{\hbar}\int_0^t  \tr\ |[V*(\rho_s-\widetilde{\rho}_s),\widetilde{\omega}_{N,s}]|\,ds+\dfrac{1}{\hbar}\int_0^t\tr\ |B_s|\,ds\,.\ee
From Proposition \ref{prop:1term-tr} and Proposition \ref{prop:2term-tr}, we get the following estimate on the terms on the r.h.s. of Eq. \eqref{eq:gronwall}:
\be\label{eq:gronwall2}
\tr\ |\omega_{N,t}-\widetilde{\omega}_{N,t}|\leq C\int_0^t\tr\ |\omega_{N,s}-\widetilde{\omega}_{N,s}|\,ds + CN\hbar\int_0^t\left(1+\sum_{k=1}^4 \|\widetilde{W}_{N,s}\|_{H_4^{k+2}} \right)\,ds\,.
\ee
We further observe that the quantity $\|\widetilde{W}_{N,s}\|_{H_4^{k+2}}$ is bounded uniformly in $N$ if $\|W_N\|_{H_4^{k+2}}\leq C$ for $k=1,\dots,4$. These regularity estimates can be obtained by simply adapting the proof of Lions and Perthame (cf. \cite{LP}). The details can be found in \cite{S19}.\\
Whence, by Gr\"{o}nwall lemma we have
\[
\tr\ |\omega_{N,t}-\widetilde{\omega}_{N,t}|\leq CN\hbar\left(1+\sum_{k=1}^4 \|{W}_{N}\|_{H_4^{k+2}} \right)\,,
\]
where $C=C(t,\|\widetilde{\rho}\|_{L^1},\|\nabla^2\widetilde{\rho}_s\|_{L^\infty})$. The boundedness of $\|\nabla^2\widetilde{\rho}_s\|_{L^\infty}$ follows from an adaptation of \cite{LP} (cf. \cite{S19} for the details of the proof).

%%%%%%%%%%%%%%%%%%%%%%%%%%%%%
%%%%%%%%%%%%%%%%
%\section{Conclusive remarks}\label{sect:conclusions}
%
%current state of art: what it is left to do.
%
%1. general mixed states with singular interaction potential
%
%2. pure states
%
%3. general pure states in the many-body picture
%
%4. mixed states with Coulomb interaction

%%%%%%%%%%%%%%%%%%%%%%%%%%%%%%%%%%%%%%%%%%%%%

%\section*{Appendix}

%%%%%%%%%%%%%%%%%%%%%%%%%%%%%%%%%%%%%%%%%%%%%
%%%%%%%%%%%%%%%%%%%%%%%%%%%%%%%%%%%%%%%%%%%%

%
\begin{acknowledgement}
The author acknowledges the support of the Swiss National Science Foundation through the Eccellenza project PCEFP2\_181153 and of the NCCR SwissMAP.
\end{acknowledgement}

\end{document}